\documentclass[letterpaper, 10pt, conference]{ieeeconf}

\IEEEoverridecommandlockouts 
\usepackage{booktabs}
\newcommand\copyrighttext{%
	\footnotesize \textcopyright 2026 IEEE. Personal use of this material is
	permitted. Permission from IEEE must be obtained for all other uses, in any
	current or future media, including reprinting/republishing this material
	for advertising or promotional purposes, creating new collective works, for
	resale or redistribution to servers or lists, or reuse of any copyrighted
	component of this work in other works.
}
\newcommand\copyrightnotice{%
	\tikzexternaldisable
	\begin{tikzpicture}[remember picture,overlay]
		\node[anchor=south,yshift=10pt] at (current page.south)
			{\fbox{\parbox{\dimexpr\textwidth-\fboxsep-\fboxrule\relax}{\copyrighttext}}};
	\end{tikzpicture}%
	\tikzexternalenable
}

\usepackage{csquotes}

\usepackage[acronym, nowarn, shortcuts]{glossaries}

\newacronym[
    plural=PDEs,
    firstplural=partial differential equations,
]{pde}{PDE}{partial differential equation}

\newacronym[
    plural=ODEs,
    firstplural=ordinary differential equations,
]{ode}{ODE}{ordinary differential equation}

\newacronym[
    plural=DDEs,
    firstplural=delay differential equations,
]{dde}{DDE}{delay differential equation}

\newacronym[
    plural=FDE,
    firstplural=functional differential equations,
]{fde}{FDE}{functional differential equation}

\newacronym[
    plural=ICs,
    firstplural=initial conditions,
]{ic}{IC}{initial condition}

\newacronym[
    plural=DPS,
    firstplural=distributed parameter systems,
]{dps}{DPS}{distributed parameter system}

\newacronym{siso}{SISO}{single-input single-output}
\newacronym{mimo}{MIMO}{multiple-input multiple-output}

\newacronym{cf}{CF}{controller form}
\newacronym{hcf}{HCF}{hyperbolic controller form}
\newacronym{mhcf}{MHCF}{MIMO hyperbolic controller form}

\newacronym{ccf}{CCF}{controller canonical form}
\newacronym{hccf}{HCCF}{hyperbolic controller canonical form}
\newacronym{nlhccf}{NL-HCCF}{nonlinear hyperbolic controller canonical form}
\newacronym{lccf}{L-CCF}{linear controller canonical form}
\newacronym{nlccf}{NL-CCF}{nonlinear controller canonical form}

\newacronym{lhocf}{L-HOCF}{linear hyperbolic observer canonical form}
\newacronym{nlhocf}{NL-HOCF}{nonlinear hyperbolic observer canonical form}
\newacronym{locf}{L-OCF}{linear observer canonical form}
\newacronym{nlocf}{NL-OCF}{nonlinear observer canonical form}

\newacronym{lhcycf}{L-HCyCF}{linear hyperbolic controllability canonical form}
\newacronym{nlhcycf}{NL-HCyCF}{nonlinear hyperbolic controllability canonical form}
\newacronym{lcycf}{L-CyCF}{linear controllability canonical form}
\newacronym{nlcycf}{NL-CyCF}{nonlinear controllability canonical form}

\newacronym{lhoycf}{L-HOyCF}{linear hyperbolic observability canonical form}
\newacronym{nlhoycf}{NL-HOyCF}{nonlinear hyperbolic observability canonical form}
\newacronym{loycf}{L-OyCF}{linear observability canonical form}
\newacronym{nloycf}{NL-OyCF}{nonlinear observability canonical form}

\newacronym[
    plural=BCs,
    firstplural=boundary conditions,
]{bc}{BC}{Boundary Condition}

\newacronym{rhs}{RHS}{right hand-side}
\newacronym{lccm}{LCCM}{leading column coefficient matrix}
\newacronym{qpld}{QPLD}{quasipolynomial long division}
\newacronym{qpldm}{MQPLD}{matrix-valued quasipolynomial long division}

\makeglossaries

\usepackage{bm}
\usepackage{amsmath}
\usepackage{amsfonts}
\usepackage{amssymb}
\usepackage{mathtools}

\newtheorem{thm}{Theorem}
\newtheorem{lem}[thm]{Lemma}

\newcommand*{\scale}[2][4]{\scalebox{#1}{$#2$}}%

\newcommand{\pderi}[1][non]{\ensuremath{\ifthenelse{\equal{#1}{non}}{\partial}}{\partial_{#1}}}

\newcommand{\trans}{^\intercal}

\newcommand{\dop}[1]{\,\mathrm{d}#1}	%

\DeclareMathSymbol{\shortminus}{\mathbin}{AMSa}{"39}

\newcommand{\diminput}{{m}}
\newcommand{\dimode}{{n}}
\newcommand{\dimfwd}{{n^-}}
\newcommand{\dimbwd}{{n^+}}

\newcommand{\kron}{\nu}
\newcommand{\kronsum}{\nu_{\Sigma}}

\newcommand{\totalshift}{\tau_{\Sigma}}

\newcommand{\lapsym}{s}
\newcommand{\diffsym}{{\mathrm{d}_t}}
\newcommand{\shiftsym}{\sigma}
\newcommand{\delaysym}{\delta}
\newcommand{\predsym}{\rho}
\newcommand{\delaymat}{\mat\Delta}
\newcommand{\predmat}{\mat P}

\newcommand{\transmatsym}{g}
\newcommand{\transmat}{\mat G}
\newcommand{\mtransmatsym}{\tilde g}
\newcommand{\mtransmat}{\mat{\tilde G}}
\newcommand{\rtransmatsym}{\bar g}
\newcommand{\rtransmat}{\mat{\bar G}}
\newcommand{\ltransmat}{\mat{\hat G}}  %
\newcommand{\ttransmat}{\mat{\check G}}  %

\newcommand{\boundarymat}{\mat H}

\newcommand{\mat}[1]{\ensuremath{\mathbf{#1}}}
\newcommand{\vecorcomp}[2][non]{%
	\ensuremath{%
		\ifthenelse{\equal{#1}{non}}{%
			\bm{#2}
		}{%
			{#2}_{#1}
		}%
	}
}

\DeclareMathOperator{\diag}{diag}
\DeclareMathOperator{\rank}{rank}

\DeclareMathOperator{\lccm}{lccm}

\DeclareMathOperator{\cdeg}{coldeg}
\DeclareMathOperator{\rdeg}{rowdeg}

\newcommand{\reals}{\mathbb{R}}

\newcommand{\realpolys}{\mathbb{R}[\diffsym]}
\newcommand{\realrationals}{\mathbb{R}(\diffsym)}
\newcommand{\ggen}{%
	\Sigma
}

\newcommand{\realgppolys}{%
	\mathcal{R}_{\mathrm{p}}%
}
\newcommand{\realgentirepolys}{%
	\mathcal{R}_{\mathrm{e}}%
}
\newcommand{\realgwholepolys}{\realgentirepolys}

\newcommand{\realgratpolys}{%
	\mathcal{R}_{\mathrm{r}}%
}

\newcommand{\leb}[1][non]{\ensuremath{%
	\ifthenelse{\equal{#1}{non}}{L^2}{L^2_{\mathrm{#1}}}
}}
\newcommand{\sob}[1][non]{H^{#1}}
\newcommand{\entirefuncs}{\mathcal{E}'}

\newcommand{\tvar}{V}
\newcommand{\boundvar}{\mathrm{BV}}
\newcommand{\boundvarleft}{\boundvar^\star}

\newcommand{\Cont}[1][non]{\ensuremath{\ifthenelse{\equal{#1}{non}}{C}{C^{\ttext{#1}}}}}

\newcommand{\leftbound}{{l_1}}
\newcommand{\rightbound}{{l_2}}
\newcommand{\mass}{{M}}

\newcommand{\oinputsym}{u}
\newcommand{\oinput}{\bm\oinputsym}

\newcommand{\odesym}{\xi}
\newcommand{\odestate}{\bm{\odesym}}

\newcommand{\pdesym}{\ensuremath{x}}
\newcommand{\pdestate}{\ensuremath{\bm{\pdesym}}}

\newcommand{\nfbasesym}{\ensuremath{\eta}}
\newcommand{\nfodess}{\ensuremath{\mathrm{O}}}
\newcommand{\nfpdess}{\ensuremath{\mathrm{P}}}

\newcommand{\nfodesym}{\nfbasesym^\nfodess}
\newcommand{\nfodesymdt}{\dot{\nfbasesym}^\nfodess}
\newcommand{\nfodestate}{\bm{\nfbasesym}^\nfodess}

\newcommand{\nfspatsym}{\ensuremath{\tau}}
\newcommand{\nfpdesym}{\ensuremath{\nfbasesym^\nfpdess}}

\newcommand{\nfstate}{\bm{\nfbasesym}}

\newcommand{\nfleftbound}[1]{0}
\newcommand{\nfrightbound}[1]{\hat\nfspatsym_{#1}}

\newcommand{\transstatefwd}{\ensuremath{\pdestate^-}}
\newcommand{\transstatebwd}{\ensuremath{\pdestate^+}}

\newcommand{\flatsym}{y}
\newcommand{\flatout}[1][non]{\vecorcomp[#1]{\flatsym}}

\newcommand{\hcfinputmap}{\mat B^\star}
\newcommand{\stieltjes}[2]{\left(#1 #2\right)(t)}

\newcommand{\heavyside}{\ensuremath{H}}

\newcommand{\degree}[1][non]{\ensuremath{\ifthenelse{\equal{#1}{non}}{^\circ}{#1^\circ}}}

\newcommand{\FT}[1][non]{\ensuremath{\mathfrak F\ifthenelse{\equal{#1}{non}}{}{ \left\{ #1 \right\}}}}
\newcommand{\IFT}[1][non]{\ensuremath{\mathfrak F^{-1}\ifthenelse{\equal{#1}{non}}{}{ \left\{ #1 \right\} }} }

\newcommand{\atan}[1][non]{\arctan\ensuremath{\ifthenelse{\equal{#1}{non}}{}{ \left( #1 \right)}}}
\newcommand{\E}[1][non]{\ensuremath{\ifthenelse{\equal{#1}{non}}{E}{E\left\{ #1 \right\}}}}
\newcommand{\chis}[2][non]{\ensuremath{\chi^2 \ifthenelse{\equal{#1}{non}}{}{ \left(#1,#2\right)}}}
\newcommand{\R}[2][xx]{\ensuremath{R_{#1}\ifthenelse{\equal{#2}{}}{}{\left( #2 \right)}}} %
\renewcommand{\S}[2][xx]{\ensuremath{S_{#1}\ifthenelse{\equal{#2}{}}{}{\left( #2 \right)}}} %
\newcommand{\Gammaf}[1][non]{\ensuremath{\Gamma\ifthenelse{\equal{#1}{non}}{}{ \left( #1 \right)}}}

\usepackage{graphicx}
\graphicspath{{./figures/}{./common/figures/}}

\usepackage{tikz}
\usetikzlibrary{%
	arrows,
	calc,
	external,
	positioning,
}
\AddToHook{env/algorithmic/begin}{\tikzexternaldisable}

\usepackage{adjustbox}

\definecolor{mplblue}{RGB}{31, 119, 180}
\definecolor{mplorange}{RGB}{255, 127, 14}
\definecolor{mplgreen}{RGB}{44, 160, 44}
\definecolor{mplred}{RGB}{214, 39, 40}
\definecolor{mplpurple}{RGB}{148, 103, 189}
\definecolor{mplbrown}{RGB}{140, 86, 75}
\definecolor{mplpink}{RGB}{227, 119, 194}
\definecolor{mplgray}{RGB}{127, 127, 127}
\definecolor{mplolive}{RGB}{188, 189, 34}
\definecolor{mplcyan}{RGB}{23, 190, 207}

\definecolor{umitred}{RGB}{128, 19, 50}
\definecolor{umitblue}{RGB}{0, 53, 103}
\definecolor{umitlightblue}{RGB}{0, 97, 163}
\definecolor{umitgreen}{RGB}{0, 92, 100}
\definecolor{umitlightgreen}{RGB}{115, 167, 156}
\definecolor{umitpurple}{RGB}{85, 25, 95}
\definecolor{umitgold}{RGB}{171, 141, 37}
\definecolor{umitlightgold}{RGB}{200, 183, 118}
\definecolor{umitlightgray}{RGB}{217, 218, 219}
\definecolor{umitgray}{RGB}{143, 143, 141}
\definecolor{umitdarkgray}{RGB}{90, 79, 74}

\usepackage{hyperref}
\usepackage{cleveref}
\crefname{app}{Appendix}{Appendeces}
\crefname{lem}{Lemma}{Lemmas}
\crefname{cor}{Corollary}{Corollaries}
\crefname{rem}{Remark}{Remarks}
\crefname{alg}{Algorithm}{Algorithms}

\usepackage{algorithm}
\usepackage{algpseudocodex}

\title{\LARGE\bf
	On the controller form for linear hyperbolic 
	MIMO systems with dynamic boundary conditions%
} 

\author{%
	Stefan Ecklebe$^{1}$ 
	and Frank Woittennek$^{1}$%
	\thanks{%
		*This research was funded in whole or in part by the Austrian Science
		Fund (FWF) [10.55776/I6519].
		For open access purposes, the author has applied a CC BY public
		copyright license to any author accepted manuscript version arising
		from this submission.
	}%
	\thanks{%
		$^{1}$	
		Institute of Automation and Control Engineering,
		UMIT TIROL, 6060 Hall in Tirol, Austria,
		{\tt\small
			\{stefan.ecklebe, frank.woittennek\}@umit-tirol.at
		}
	}%
}

\begin{document}

\maketitle
\thispagestyle{empty}
\copyrightnotice
\pagestyle{empty}
\begin{abstract}
	This contribution develops an algebraic approach to obtain a
controller form for a class of linear hyperbolic MIMO systems,
bidirectionally coupled with a linear ODE system at the unactuated boundary.
After a short summary of established controller forms for SISO and MIMO ODE as
well as SISO hyperbolic PDE systems, it is shown that the approach to state a
controller form for SISO systems cannot easily be transferred to the MIMO case
as it already fails for a very simple example.
Next, a generalised hyperbolic controller form with different variants is
proposed and a new flatness-based scheme to compute said form is presented. 
Therein, the system is treated in an algebraic setting
where quasipolynomials are used to express the predictions and delays in the
system.
The proposed algorithm is then applied to the motivating example.

\end{abstract}

\section{Introduction}
Canonical forms serve as a cornerstone for systematic feedback
design in finite dimensions. Their counterparts for \glspl{dps}
have been studied since the late 1970s. 
For linear hyperbolic \glspl{pde}, \cite{Russel1978}
connected controllability and spectral assignment to
\glspl{ccf} and neutral delay equations.  
More recent works linked the aforementioned delay interpretations to flatness
based parametrisations \cite{FMRR95cdc,FliMou99} (as well as
\cite{CookeKrumme1968jmaa} for early ideas in this direction) and established
the connection between such parametrisations and the \gls{ccf} for more
general systems.
In particular, \cite{Woittennek2012-ifacwc} introduced a \gls{hccf} for linear
hyperbolic SISO systems comprising general interconnections of hyperbolic
\glspl{pde} and \glspl{ode} (see also \cite{Woittennek2011ifac}) by generalising
the familiar chain-of-integrators structure by prepending a transport equation
to represent the delays of the system imposed by the hyperbolic \glspl{pde}.
Therein, the canonical coordinates are deduced from the feasible \glspl{ic} for
a \gls{fde} arising from a flatness-based parametrisation of the system's input
while the required state transforms are computed using the flatness-based
parametrisation of the system's original state variables.
Aside from the above-mentioned state-space approaches, in  \cite{Woittennek2010-siam,Gehring2013}  flatness-based parametrisations also served as means for purely algebraic analysis and control design. These methods employ and partially generalise results from
\cite{Brethe1996} and \cite{GluesingLuerssen1997}.

As many relevant technological applications have multiple in- and outputs,
several control approaches are available, e.g.~ \cite{Knueppel2014} using
flatness-based methods or \cite{Gabriel2024} employing backstepping designs.
However, regarding controller forms a direct extension of the existing work to
the \acs{mimo} case is difficult, as such systems can exhibit
different delays depending on the considered input which together with the
bidirectionally coupled ODE result in an \gls{fde} system for which feasible
\glspl{ic} are not obvious, thus, negating the approach from the \acs{siso} case.

Therefore, the goal of this contribution is to introduce a \gls{hcf} for a class
of linear hyperbolic \acs{mimo} systems with dynamic boundary conditions and to
develop a delay reduction algorithm such that the resulting
\gls{fde} allows for a direct computation of the controller form.
To do so, we denote the input-output relation using quasipolynomials,
whose use
in the control of hyperbolic systems was already discussed in
\cite{Loiseau1998,Benarab2022},
and restate the requirements to introduce a \gls{hcf}
as column-reduced properties on the resulting quasipolynomial matrices.
Next, in contrast to the established methods and due to the incommensurate
delays in the considered systems as well as the interplay with the boundary
system, we propose a modified quasipolynomial division algorithm to compute
the required transformation to establish said properties.

Following some preliminaries, where we briefly restate established controller 
forms and demonstrate the problems that arise if a hyperbolic controller form
is to be introduced even for the most simple \acs{mimo} system,
a general \acs{mimo} \gls{hcf} and several subforms are given in \Cref{sec:mimo_hcf}.
The main part of the contribution in \Cref{sec:computation} then discusses
how a \gls{hcf} can be obtained for a specific class of systems 
Finally, \Cref{sec:example} traces the procedure for an example system.

\section{Preliminaries}

\subsection{Notation}

The variables $z$ and $t$ denote space and time.
The first and second derivative of $x(t)$ is denoted
by $\dot x(t)$ and $\ddot x(t)$, while the $n$-th 
derivative is given by $x^{(n)}(t)$
and the partial derivative of $x(z,t)$ w.r.t.~$z$
by $\partial_z x(z,t)$.

To treat \glspl{fde} in an algebraic setting, we identify the
operators therein with their Fourier-Laplace transforms 
and treat the application of the transformed operators
as convolutions.
Thus, we write $\mathcal{L}\{\diffsym\} = \lapsym$ for
differentials, while  
the shifts
$(\shiftsym^\tau x)(t) = x(t+\tau)$ 
with $\tau \in \reals$
are denoted as
$\mathcal{L}\{\shiftsym^\tau\}(\lapsym) = \exp(\lapsym\tau)$.

The shorthand $\leb([a,b])$ is used for the space of square integrable
functions on $[a,b]$, while $\sob[n]([a,b])$ stands for the space of functions whose $n$-th
derivative belongs to $\leb([a,b])$
and $\entirefuncs$
for Schwartz-Distributions with compact support.
With $\tvar_f([a,b])$ denoting the total variation of a function $f$ on the
interval $[a,b]$, $\boundvar([a,b])$ contains functions of bounded
variation while members of its subset $\boundvarleft([a,b])$ 
satisfy
$\lim_{\epsilon \to 0} \tvar_f([b,b-\epsilon)) = 0$.

Furthermore, $\realpolys$ is the space of polynomials in $\diffsym$ with coefficients in $\reals$
while $\realrationals$ consists of rational expressions in $\diffsym$. 
Regarding a polynomial matrix $\mat X(\diffsym) \in \reals^{n\times n}[\diffsym]$, we denote the \gls{lccm} w.r.t.~$\diffsym$ by $\lccm_{\diffsym}{\mat X(\diffsym)}$ and say that $\mat X(\diffsym)$ is column reduced w.r.t.~$\diffsym$ if $\lccm_{\diffsym}{\mat X(\diffsym)}$ has full generic rank.
Finally, we call $\mat X(\diffsym)$ unimodular over $\realpolys$ if its determinant is a nonzero constant (and therefore independent of $\diffsym$).

\subsection{Established controller forms}%

\subsubsection{ODE SISO systems}%

For this class, \cite{Kalman1963} defines the controller canonical form
\begin{subequations}
	\begin{align}
		\nfodesymdt_{1}(t) &= \nfodesym_2(t)\\
		&\;\;\vdots\nonumber\\
		\nfodesymdt_{n-1}(t) &= \nfodesym_n(t)\\
		\nfodesymdt_{n}(t) &= -\bm{a}\trans\nfodestate(t) + \oinputsym(t)
	\end{align}%
	\label{eq:siso_cf}%
\end{subequations}
with state 
$\nfodestate(t) 
= (\nfodesym_1(t),\dotsc,\nfodesym_n(t))\trans \in \reals^\dimode$,
coefficients $\bm a \in \reals^n$
and input $\oinputsym(t) \in\reals$, which exists for every controllable
system.
Note that the \enquote{output} $\nfodesym_1(t)$ of the integrator chain
is a flat output of the system.

\subsubsection{ODE MIMO systems}%
\label{sss:lumped_mimo}

Extending the \acs{siso} framework to systems with
inputs $\oinput(t) \in \reals^m$, \cite{Luenberger1967} 
states a controller form\footnote{%
	As these forms are not unique in the \acs{mimo} case, 
	the ancillary adjective canonical will be omitted from now on.
	Moreover, note that \eqref{eq:mimo_cf_input} explicitly allows
	for all inputs to occur in the \enquote{input} of the subsystems
	integrator chain with linear factors $\neq 1$ as long as
	$\mat B = (\bm b_1,\dotsc,\bm b_m)\trans$ has full rank.
}, with subsystems
\begin{subequations}
	\begin{align}
		\nfodesymdt_{i,1}(t) &= \nfodesym_{i,2}(t)\\
		&\;\;\vdots\nonumber\\
		\nfodesymdt_{i,\kron_i-1}(t) &= \nfodesym_{i,\kron_i}(t)\\
		\label{eq:mimo_cf_input}%
		\nfodesymdt_{i,\kron_i}(t) &= 
		-\sum_{j=1}^{m} \bm{a}_{ij}\trans\nfodestate_j(t) 
			+ \bm{b}_i\trans\oinput(t)
	\end{align}
	\label{eq:mimo_cf}%
\end{subequations}
for $i\in\{1,\dotsc,\diminput\}$,
with the Kronecker indices $\kron_i$ and
the state vector 
$\nfodestate_i(t) = 
	(\nfodesym_{i,1}(t),
	\dotsc,
	\nfodesym_{i,\kron_i}(t))\trans 
	\in \reals^{\kron_i}$
of the $i$-th subsystem.
Similarly to the \acs{siso} case, the \acs{mimo} 
controller form state $\nfodestate(t) \in \reals^{\kron_\Sigma}$
with $\kronsum \coloneq \sum_{i=1}^{m}\kron_i$ 
and input $\oinput(t)$
can be easily parametrised by a flat output 
$(\nfodesym_{1,1}(t),\dotsc,\nfodesym_{m,1}(t))$.

Note that, in principle, there exist several ways to distribute the $\kron_\Sigma$
integrators over the $\diminput$ chains but only one that yields
the form \eqref{eq:mimo_cf} without input derivatives.

\subsubsection{Hyperbolic SISO systems with boundary dynamics}%

Initially, the work \cite{Russel1978} introduced a controller
form for boundary controlled hyperbolic systems 
that takes the form of decoupled transport systems.
These results were then extended to dynamic boundary
conditions in \cite{Woittennek2012-ifacwc,Woittennek2013-cpde},
yielding the controller form
\begin{subequations}
	\begin{align}
		\nfodesymdt_{1}(t) &= \nfodesym_2(t)\\
		&\;\;\vdots\nonumber\\
		\nfodesymdt_{n-1}(t) &= \nfodesym_n(t)\\
		\nfodesymdt_{n}(t) &= \nfpdesym(\nfleftbound{},t) \\
		\partial_t \nfpdesym(\nfspatsym,t) &= 
			\partial_\nfspatsym \nfpdesym(\nfspatsym,t),
			\quad \nfspatsym \in [\nfleftbound{},\nfrightbound{})\\
		\nfpdesym(\nfrightbound{},t) &= 
		- \stieltjes{a^\star}{\nfpdesym}
			- \bm{a}\trans\nfodesym(t)
			+ \oinputsym(t)
	\end{align}%
	\label{eq:siso_hcf}%
\end{subequations}
with state 
$\nfstate(\cdot,t)
= (\nfodestate(t), \nfpdesym(\cdot,t)) 
\in 
\reals^n 
\times 
\leb([\nfleftbound{},\nfrightbound{}])
$
and the Stieltjes integral
$
\stieltjes{a^\star}{\nfpdesym}
=\int_{\nfleftbound{}}^{\nfrightbound{}}\nfpdesym(\nfspatsym,t)\dop \bar a(\nfspatsym)
$
where $\bar a \in BV^\star([\nfleftbound{},\nfrightbound{}])$
allows for distributed feedback of $\nfpdesym(z,t)$
as well as point evaluations for $z\in[\nfleftbound{},\nfrightbound{})$.
Similar to the SISO case, $\nfodesym_1(t)$ is a flat output of the system.

\subsection{Benchmark system}%
\label{ssc:benchmark}

Consider a system of two vibrating strings
\begin{subequations}
	\begin{align}
		\partial_t^2 x_1(z,t) - \partial_z^2 x_1(z,t) &= 0\quad z \in (-\leftbound, 0)\\
		\partial_t^2 x_2(z,t) - \partial_z^2 x_2(z,t) &= 0\quad z \in (0, \rightbound),
	\end{align}%
	\label{eq:zuazua_pdes}%
\end{subequations}
coupled with a point mass $\mass$ at $z=0$ and thus
subject to the dynamic boundary conditions
\begin{subequations}
	\begin{align}
		\label{eq:zuazua_bc_diric}
		x_1(0,t) = x_2(0,t) &= \odesym(t)\\
		\label{eq:zuazua_bc_dyn}
		\mass \ddot\odesym(t) &= -\partial_z x_1(0,t) + \partial_z x_2(0,t),
	\intertext{where $\odesym(t)$ denotes the mass position as well as}
		\label{eq:zuazua_bc_inp1}
		\partial_z x_1(-\leftbound,t) &= -u_1(t)\\
		\label{eq:zuazua_bc_inp2}
		\partial_z x_2(\rightbound,t) &= u_2(t),
	\end{align}%
	\label{eq:zuazua_bcs}%
\end{subequations}
with state\footnote{%
	In detail, 
	$\bm x(\cdot,t) \!=\! (
	\odesym(t),\dot \odesym(t),
	x_1(\cdot,t), \partial_t x_1(\cdot,t),
	x_2(\cdot,t), \partial_t x_2(\cdot,t)
	)\trans$
	and
	$ \mathcal{X} = \{
	\reals^2 \times
	\sob[1]([-\leftbound,0]) 
	\times
	\leb([-\leftbound,0])
	\times
	\sob[1]([0,\rightbound]) 
	\times
	\leb([0,\rightbound])
	\;\vert\;
	\odesym(t) = x_1(0,t) = x_2(0,t)
	\}$.
}
$\bm x(\cdot,t) \in \mathcal{X}$ as
initially discussed in \cite{Hansen1995}.

\subsection{Input parametrisation by a flat output}%
\label{sub:Parametrisation}

Starting at the interior mass, by using the tuple
\begin{equation}
	\begin{pmatrix}
		\flatout[1](t),
		\flatout[2](t)
	\end{pmatrix}\trans
	= 
	\begin{pmatrix}
		\odesym(t), 
		\partial_z x_1(0,t) + \partial_z x_2(0,t)
	\end{pmatrix}\trans
	\label{eq:flat_out}%
\end{equation}
with $y_1 \in \sob[2](\reals)$ as well as
$ y_2 \in \leb(\reals) $
we directly obtain $\odesym(t) = x_1(0,t) = x_2(0,t) = \flatout[1](t)$
and via \eqref{eq:zuazua_bc_dyn}
\begin{subequations}
	\begin{align}
		\partial_z x_1(0,t) &= \frac{1}{2}\left(\flatout[2](t) - \mass \diffsym^2\flatout[1](t)\right)\\
		\partial_z x_2(0,t) &= \frac{1}{2}\left(\flatout[2](t) + \mass \diffsym^2\flatout[1](t)\right).
	\end{align}%
	\label{eq:zuazua_param_xdz0}%
\end{subequations}
By using the deflection $x_i(0,t)$ and its gradient $\partial_z x_i(0,t)$
as Cauchy data, 
the Method of Characteristics
yields the profiles
\begin{equation}
	x_i(z,t) \!
		=
		\!\frac{1}{2} \!\left(
			\!
			\odesym(t-z) 
			\!+\!\odesym(t+z)
		\!
		+
		\!
		\!
		\int_{t-z}^{t+z} 
			\!
			\partial_z x_i(0, \nfspatsym) \dop{\nfspatsym} 
		\right)
	\label{eq:zuazua_param_x}%
\end{equation}
for $i=1,2$ (cf.~\cite{Hansen1995,Woittennek2012-ifacwc}),
which can be written as
\begin{equation}
	x_i(z,t) = 
		C_z x_i(0,t) 
		+ S_z \partial_z x_i(0,t)
	\label{eq:profile_sol_alg}%
\end{equation}
in the algebraic setting,
where 
$C_z 
	= \frac{1}{2}(\shiftsym^z+\shiftsym^{-z})$
and 
$S_z 
	= \frac{1}{2}\frac{\shiftsym^z-\shiftsym^{-z}}{\diffsym}
$
employ $\shiftsym$ as well as $\diffsym$ for shifts and derivatives.

Herein, note that although $S_z$ contains a division by
$\diffsym$ which results in an integration in the parametrisation of
$x_i(z,t)$ by $\flatout(t)$, this does not contradict the idea
of a unique parametrisation as $S_z \in \entirefuncs$ has compact support and
the corresponding integral in \eqref{eq:zuazua_param_x} therefore is definite.

With the parametrisation of the Cauchy data from \eqref{eq:zuazua_param_xdz0},
\eqref{eq:profile_sol_alg} then yields the profiles, thus, the input
parametrisation
\begin{equation}
	\begin{pmatrix}
		\oinputsym_1(t)\\
		\oinputsym_2(t)\\
	\end{pmatrix}
	= 
	\begin{pmatrix}
		\frac{m}{2} \diffsym^2 C_\leftbound + \diffsym S_\leftbound
			& -\frac{1}{2} C_\leftbound\\
		\frac{m}{2} \diffsym^2 C_\rightbound + \diffsym S_\rightbound
			& \frac{1}{2} C_\rightbound
	\end{pmatrix}
	\begin{pmatrix}
		\flatsym_1(t)\\
		\flatsym_2(t)\\
	\end{pmatrix}
	\label{eq:zuazua_param_u}%
\end{equation}
is obtained via \eqref{eq:zuazua_bc_inp1} and \eqref{eq:zuazua_bc_inp2}
showing that \eqref{eq:flat_out} is a flat output of the system 
\eqref{eq:zuazua_pdes}-\eqref{eq:zuazua_bcs}.

\subsection{Stating a controller form}%

As known from e.g.~\cite{Woittennek2012-ifacwc}, the problem
to be solved for stating a controller form is to find
a particular realisation of the input-output relation
\eqref{eq:zuazua_param_u}.
A possible way to do so is to solve \eqref{eq:zuazua_param_u} for the highest
derivatives of the largest predictions of $\flatsym_1(t)$ as well as
$\flatsym_2(t)$ and to check for which initial conditions the resulting
\gls{fde} system is well-posed.

However, the leading column coefficient matrix
\begin{equation}
	\label{eq:zuazua_lccm}
	\hat\transmat =
	\frac{1}{4}
	\begin{pmatrix}
		0 & 0\\
		\mass & 1\\
	\end{pmatrix},
\end{equation}
obtained by taking the coefficients of the highest derivatives of the largest
predictions in each column on the right hand-side of \eqref{eq:zuazua_param_u}, 
is singular.
This implies that \eqref{eq:zuazua_param_u} cannot directly be written as an
explicit \gls{fde} system thus negating our approach.
Therefore, more steps must be taken in the \acs{mimo} case, 
which will be discussed in \Cref{sec:computation} after we give
a definition of a controller form for this system class.

\section{A hyperbolic controller form}%
\label{sec:mimo_hcf}

In this section we define a general controller form for linear hyperbolic \acs{mimo} systems with dynamic boundary conditions and state several specific variants of the latter which may be helpful in different settings.

\subsection{Definition}%
\label{sub:Definition}

By directly applying the extension of the \acs{siso} controller form  by a
transport system from \eqref{eq:siso_hcf} to the \acs{mimo} controller form
\eqref{eq:mimo_cf}, for $i\in\{1,\dotsc,m\}$, we obtain the subsystems
\begin{subequations}
	\begin{align}
		\nfodesymdt_{i,1}(t) &= \nfodesym_{i,2}(t)\\
		&\;\;\vdots\nonumber\\
		\nfodesymdt_{i,\kron_i-1}(t) &= \nfodesym_{i,\kron_i}(t)\\
		\nfodesymdt_{i,\kron_i}(t) &= \nfpdesym_{i}(\nfleftbound{i},t) \\
		\partial_t \nfpdesym_{i}(\nfspatsym,t) &= 
		\partial_\nfspatsym \nfpdesym_{i}(\nfspatsym,t),
		\quad \nfspatsym \in [\nfleftbound{i},\nfrightbound{i})\\
			\nfpdesym_{i}(\nfrightbound{i},t) &= 
			-\sum_{j=1}^{m}\Big(
				\stieltjes{a^\star_{ij}}{\nfpdesym_{j}}
				+ \bm{a}_{ij}\trans\nfodesym_{j}(t) 
				-\stieltjes{b_{ij}^\star}{\oinputsym_j}
			  \Big)
		\label{eq:mimo_hcf_input}%
	\end{align}%
	\label{eq:mimo_hcf}%
\end{subequations}
with
state 
$
\nfstate(\cdot,t) \in 
\reals^{\kron_\Sigma}
\times
\prod_{i=1}^m \leb([\nfleftbound{i},\nfrightbound{i}])
$.
Herein, the expressions
$
\stieltjes{b_{ij}^\star}{\oinputsym_j}
= \int_{\underline{\nfspatsym}_{ij}\!}^{\bar\nfspatsym_{ij}} 
	\oinputsym_j(t+\nfspatsym)\dop \bar b_{ij}(\nfspatsym)
$
correspond to convolutions that -- depending on the $\bar b_{ij}(\tau)$ --
can represent input derivatives, predictions or delays.
Furthermore, the input map 
$\mat B^\star = (b^\star_{ij})$
is invertible w.r.t.~convolution. 

In \eqref{eq:mimo_hcf}, each block consist of a transport 
system whose output feeds into an integrator chain.
However, the input of each block's transport system now depends on the system
input as well as the \acs{pde} and \acs{ode} states of all blocks.

Note that similar to the integrator chains from the lumped MIMO case in
\cref{sss:lumped_mimo}, there exist different ways to split the overall system
delays into the $m$ transport systems.

In the following, we will derive selected special cases of \eqref{eq:mimo_hcf}
by further specifying the functions $\bar b_{ij}(\nfspatsym)$.

\subsection{A classic controller form}%

To recover a classic form without input derivatives, predictions or delays
in $\hcfinputmap$,
we restrain the functions to $\bar b_{ij}(\nfspatsym) = b_{ij}\heavyside(\nfspatsym)$
with the Heaviside step function $\heavyside(\nfspatsym)$, resulting in 
\begin{equation}
	\label{eq:chcf}
	\sum\limits_{j=1}^{m} 
		\stieltjes{b^\star_{ij}}{\oinputsym_{j}}
	= \bm{b}_i\trans \oinput(t)
\end{equation}
for $i \in \{1,\dotsc,\diminput\}$ after evaluating the integrals.

\subsection{Generalised controller forms}%
\label{ssc:generalised_forms}

\subsubsection{With input derivatives}
\label{sss:gform_deriv}

To allow for input derivatives to occur in \eqref{eq:mimo_hcf_input},
we set $\bar b_{ij}(\nfspatsym) = \sum_{k=0}^{D_{ij}} b_{ij,k}\heavyside^{(k)}(\nfspatsym)$,
yielding
\begin{equation}
	\label{eq:ghcf_input_deriv}
	\sum\limits_{j=1}^{m}
	\stieltjes{b^\star_{ij}}{\oinputsym_{j}}
	= \sum\limits_{j=1}^{m}
	\sum_{k=0}^{D_{ij}} b_{ij,k} \oinputsym_{j}^{(k)}(t)
\end{equation}
where the $D_{i,j}$ denote the highest derivative order of 
$\oinputsym_j(t)$ in the equations of subsystem $i$.
Owing to the derivatives of the $u_j \in \leb[loc](\reals^+)$, 
we propose the term \textit{discontinuous controller form} for the
subtype \eqref{eq:ghcf_input_deriv}.

\subsubsection{With input predictions}
\label{sss:gform_pred}

For this case, we have
$\bar b_{ij} \in \boundvar([0, \bar\nfspatsym_{ij}])$,
where the $\bar\nfspatsym_{ij}$ denote the biggest predictions with which
the input $\oinputsym_j(t)$ appears in the subsystem $i$.
As such a system depends on future inputs, we propose the term 
\textit{non-causal controller form} for this case.

\subsubsection{With input predictions and delays}
\label{sss:gform_pred_and_delay}

For a more general case, we allow 
$b_{ij}^\star \in \boundvar([\underline{\nfspatsym\!}_{ij}, \bar\nfspatsym_{ij}])$,
where the $\underline{\nfspatsym\!}_{ij}$ denote the biggest delays 
of input $\oinputsym_j(t)$.
As this form only yields a sensible state if the system 
is defined for $t\in\reals$ instead of $\reals^+$,
we propose the term \textit{quasi controller form}.

\section{Computation of the controller form}%
\label{sec:computation}

This section contains the main result of the contribution, 
in which we present an algorithm to compute the controller form 
for a specific class of systems.

\subsection{System class}%
\label{sub:comp_class}

We discuss the linear first order hyperbolic \acs{mimo} system\footnote{%
	Note that the system from \Cref{ssc:benchmark} can easily be
	transformed into this more general form which also allows for asymmetric
	predictions and delays in contrast to the more compact second order
	formulation.
}
\begin{subequations}%
	\begin{align}
		\label{eq:orig_ode}
		\dot{\odestate}(t) &= \bm{F}\odestate(t) 
			+ \bm{B}\transstatefwd(0,t) \\
		\label{eq:orig_pde}
		\partial_t\pdestate(z,t) &= 
			\bm{\Lambda}\partial_z\pdestate(z,t),
			\quad z \in (0,1)\\
		\label{eq:orig_bc0}
		\transstatebwd(0,t) &= \bm{Q}_0\transstatefwd(0,t)
		+ \bm{C}\bm{\xi}(t)\\
		\label{eq:orig_bc1}
		\transstatefwd(1,t) &= \bm{Q}_1\transstatebwd(1,t) 
		+ \bm{u}(t),
	\end{align}%
	\label{eq:orig_sys}%
\end{subequations}%
where $\odestate(t) \in \reals^\dimode$, 
$\transstatefwd(\cdot,t) \in \leb([0,1],\reals^{\dimfwd})$
and 
$\transstatebwd(\cdot,t) \in \leb([0,1],\reals^\dimbwd)$
denote the boundary system's as well as the backward and 
forward transport systems' states.
This model is also known as a
\acs{pde}-\acs{ode} system in the backstepping literature, cf.~i.e.~\cite{DiMeglio2018}.

For \eqref{eq:orig_sys}, we assume decoupled \acsp{pde} with
$\bm{\Lambda} = \mathrm{diag}(
	\mat\Lambda^-, -\mat\Lambda^+
	)$
where $\mat\Lambda^\pm = 
\mathrm{diag}(\lambda_1^\pm,\dotsc,\lambda_{n^\pm}^\pm)$,
$\lambda_i^{\pm} \ge \lambda_{i+1}^{\pm} > 0$,
the pair $(\bm F, \bm B)$ controllable and
$\rank \bm Q_0 = \rank \bm Q_1 = \dimbwd$.

\subsection{Parametrisation}%
\label{sub:comp_param}

As \eqref{eq:orig_ode} is controllable, a
flatness-based parametrisation 
\begin{subequations}
	\begin{align}
		\label{eq:ode_state_param}
		\odestate(t) &= \mat N(\diffsym) \flatout(t)\\
		\label{eq:ode_input_param}
		\transstatefwd(0,t) &= \mat D(\diffsym) \flatout(t)
	\end{align}%
	\label{eq:ode_param}%
\end{subequations}
can always be given. 
	Furthermore, if the components of the flat output $\flatout(t)$ are chosen
	as the outputs of each block's integrator chain of the systems' \acs{mimo}
	controller form \eqref{eq:mimo_cf}, the parametrisation matrix $\mat D(\diffsym)$
	is column reduced and $\cdeg_{i} \mat N(\diffsym) < \cdeg_{i}\mat D(\diffsym),\;i
	\in \{1,\dotsc,m\}$ holds for the column degrees.

Regarding the \acs{pde} subsystem, with $\odestate(t)$ as well as
$\transstatefwd(0,t)$ known, \eqref{eq:orig_bc0} yields
$\transstatebwd(0,t) 
= \mat\boundarymat(\diffsym)\flatout(t)
= \left(
	\mat Q_0 \mat D(\diffsym) 
	+ \mat C \mat N(\diffsym)
\right)\flatout(t)$ 
and therefore the boundary values required to solve
\eqref{eq:orig_pde} up to $z=1$. This yields
\begin{align}
	\label{eq:pde_fwd_state_param}
	\transstatefwd(1,t)
	&= \predmat(\shiftsym, t) \transstatefwd(0,t)
	\\
	\label{eq:pde_bwd_state_param}
	\transstatebwd(1,t)
	&= \delaymat(\shiftsym, \diffsym) \transstatebwd(0,t)
\end{align}
with the prediction matrix
$\predmat(\shiftsym) = \diag(
\shiftsym^{\tau_{1}^-},
\dotsc,
\shiftsym^{\tau_{\dimfwd}^-}
)$
and the delay matrix
$\delaymat(\shiftsym) = \diag(
\shiftsym^{-\tau_{1}^+},
\dotsc,
\shiftsym^{-\tau_{\dimbwd}^+}
)$
where $\tau_i^{\pm} = \frac{1}{\lambda_i^{\pm}}$. 
Finally,
combining \Cref{eq:ode_state_param,eq:ode_input_param,eq:pde_fwd_state_param,eq:pde_bwd_state_param}
with \eqref{eq:orig_bc1} yields the matrix
\begin{equation}
	\label{eq:transmat_def}
	\transmat(\shiftsym,\diffsym) 
    \!=\! \predmat(\shiftsym)\mat D(\diffsym) 
		- \mat Q_1\delaymat(\shiftsym)
		\boundarymat(\diffsym)
\end{equation}
for the parametrisation
\begin{equation}
	\oinput(t) = \transmat(\shiftsym,\diffsym)\flatout(t)
	\label{eq:io_form}
\end{equation}
of the system input by the flat output of the \acs{ode}. 

The structure of \eqref{eq:transmat_def} now shows that
$\mat \predmat(\shiftsym)$ and $\delaymat(\shiftsym)$ introduce \emph{all}
predictions and delays of the system into \emph{every} column of
$\transmat(\shiftsym,\diffsym)$ which is the root cause of the
problem with the motivational example's parametrisation
\eqref{eq:zuazua_param_u}.
Next, we will employ quasipolynomials to resolve the issue.

\subsection{Quasipolynomials}%

With the set $T = \{
\tau_1^-,\dotsc,\tau_\dimfwd^-,
\tau_1^+,\dotsc,\tau_\dimbwd^+,
\}
$
of all unique predictions and shifts in the system, by interpreting the members
of $\ggen = \{\sigma^{\tau},\sigma^{-\tau} | \tau \in T\}$ and $\diffsym$ as
independent generators, we conceive the entries $\transmatsym_{i,j}(\shiftsym,\diffsym)$
of $\transmat(\shiftsym,\diffsym)$ to be from $\realgppolys =
\mathbb{R}[\ggen,\diffsym]$, the ring of quasipolynomials in
$\shiftsym$ with real exponents, whose coefficients are polynomials in
$\diffsym$ with coefficients from $\reals$. 
Furthermore, for the elements
\begin{equation}
	r(\shiftsym,\diffsym) = \sum_{i=1}^n a_i(\diffsym)\shiftsym^{\alpha_i}
	\in \realgppolys
	\label{eq:gppoly}
\end{equation}
with $a_i(\diffsym) \in \realpolys$, the ordering $\alpha_i > \alpha_{i+1}$, 
for $i \in \{1,\dotsc,n-1\}$
is assumed.
To ease working on the polynomials \eqref{eq:gppoly}, we define the pseudo degrees
\begin{equation*}
	\deg_{\shiftsym}^+ r(\shiftsym,\diffsym) 
		\coloneq \max_{i \in \{1,\dotsc,n\}} \alpha_i, 
	\;
	\deg_{\shiftsym}^- r(\shiftsym,\diffsym) 
		\coloneq \min_{i \in \{1,\dotsc,n\}} \alpha_i 
\end{equation*}
as well as the degree\footnote{%
	This choice is motivated by the case of Laurent polynomials,
	which possess positive and negative integer exponents and
	form an euclidean ring for the given degree definition,
	cf.~\cite{GluesingLuerssen1997,Icart2012}.
}
\begin{equation}
	\deg_{\shiftsym} r(\shiftsym,\diffsym) \coloneq 
		\deg_{\shiftsym}^+ r(\shiftsym,\diffsym) - \deg_{\shiftsym}^- r(\shiftsym,\diffsym) 
		,
	\label{eq:degree_def}
\end{equation}
yielding the upper and lower boundary as well as the diameter of the interval,
respectively.
Next, we discuss how these intervals can be reduced in a systematic way.

\subsection{A shift interval reduction algorithm}%
\label{ssc:reduction_strategy}

After stating the requirements on the input parametrisation in terms
of degree conditions for quasipolynomials, this section
shows how the shifts in single entries of $\transmat(\shiftsym,\diffsym)$
can be reduced
before the element-wise operations are condensed into an algorithm to reduce
the shifts in the whole matrix.

\subsubsection{Requirements}%
\label{sss:reduction_requirements}

To obtain the controller form \eqref{eq:mimo_hcf},
the input parametrisation \eqref{eq:io_form} must be solved for the highest
derivatives of the flat output.
Furthermore, to introduce a \gls{hcf} state in a direct fashion,
these must occur together with the largest predictions as well as
largest delays and the total system shift
$ \totalshift = \sum_{i=1}^{\dimfwd} \tau^-_i 
		+ \sum_{i=1}^{\dimbwd} \tau^+_i
$
must be preserved.
In the algebraic setting, this corresponds to the conditions
\begin{subequations}
	\begin{align}
		\label{eq:req_lccm_cond}%
		\lccm_{\diffsym}(\lccm^+_{\shiftsym} \transmat(\shiftsym,\diffsym))
		&= \lccm^+_\shiftsym(\lccm_{\diffsym} \transmat(\shiftsym,\diffsym))
		\\
		\label{eq:req_state_dim}
		\sum_{i=1}^m \cdeg_{\sigma,i} \transmat(\shiftsym,\diffsym)
		&= \totalshift
	\end{align}%
	\label{eq:req_alg_conds}%
\end{subequations}
where \eqref{eq:req_lccm_cond} must have full rank. 
In this setting, the operator $\lccm^+_{\shiftsym}$ denotes the \gls{lccm}
w.r.t.~$\shiftsym$ while using the pseudo degree $\deg^+_\shiftsym$ to identify
the leading coefficients. 
Thus, \eqref{eq:req_lccm_cond} is trivially satisfied if
the degree conditions
\begin{subequations}%
	\begin{align}%
		\label{eq:req_pred_ordering}%
		\deg^+_\shiftsym {\transmatsym}_{ij}(\shiftsym,\diffsym)
		&\le 
		\deg^+_\shiftsym {\transmatsym}_{jj}(\shiftsym,\diffsym)
		\\
		\label{eq:req_derivative_ordering}
		\deg_{\diffsym} {\transmatsym}_{ij}(\shiftsym,\diffsym)
		&\le 
		\deg_{\diffsym} {\transmatsym}_{jj}(\shiftsym,\diffsym)
	\end{align}%
	\label{eq:req_degree_cond}%
\end{subequations}
hold for $i,j \in \{1,\dotsc,\diminput\}$
as all \glspl{lccm} then are diagonal.

\subsubsection{Preliminary transformations}%
\label{sss:preliminary_transforms}

\begin{lem}
	\label[lem]{lem:sorting}
	By transforming the system input and permuting the
	components of the flat output, a matrix $\mtransmat(\shiftsym,\diffsym)$
	can be obtained, such that 
	a) the leading row delays 
	$d_i = \rdeg_{\shiftsym,i}^- \mtransmat(\shiftsym,\diffsym)$
	coincide, 
	b) the leading row degrees 
	$r_i = \rdeg_{\shiftsym,i} \mtransmat(\shiftsym,\diffsym)$
	are ascending and 
	c) the highest degrees w.r.t.~$\diffsym$ in each column are
	found on the main diagonal. 
	Thus, $\mtransmat(\shiftsym,\diffsym)$ satisfies
	\eqref{eq:req_derivative_ordering}
	as well as 
	\begin{subequations}
		\begin{align}
			\label{eq:equal_leading_row_delays}
			d_1 &= d_2 = \hdots = d_m,\\
			\label{eq:shift_ordering}
			r_i &\le r_{i+1},\; i\in\{1,\dotsc,m-1\}.
		\end{align}%
		\label{eq:transmat_ordering}%
	\end{subequations}%
\end{lem}
\begin{proof}
	Due to the structure of \eqref{eq:transmat_def}, 
	a fully populated $\mat Q_1$ introduces
	larger delays from lower rows into rows further up.
	Therefore, we use a lower left triangular row transform 
	(obtainable e.g.~from the LU decomposition of $\mat Q_1$)
	to bring $\mat Q_1$ into upper left
	triangular form.
	This eliminates superfluous delays without increasing the
	predictions as $\predmat(\shiftsym)$ is diagonal.
	Next, the leading row delays are harmonised at their mean $\bar d =
	\sum_{i=1}^m d_i$ by left multiplying with the unimodular 
	$\mat S(\shiftsym) 
	= \diag(\shiftsym^{\bar d - d_1},\dotsc,\shiftsym^{\bar d - d_m})$.
	Finally, \eqref{eq:shift_ordering} is obtained by permuting the input
	components.

	Furthermore, from the parametrisation \eqref{eq:ode_param} we have that
	$\cdeg_i \mat N(\diffsym) < \cdeg_i \mat D(\diffsym)$ 
	(which is why $\mat C$ has no effect regarding the degrees in
	\eqref{eq:transmat_def} w.r.t.~${\diffsym}$)
	and also that $\mat D(\diffsym)$ is column reduced.
	Since this property is preserved under transformations over $\reals$
	and $\mat Q_0, \mat Q_1$ have full rank, we obtain that
	$\transmat(\shiftsym,\diffsym)$ is column reduced w.r.t.~$\diffsym$.
	Therefore, \eqref{eq:req_derivative_ordering} is achieved by permuting
	the components of $\flatout(t)$.
\end{proof}

In $\mtransmat(\shiftsym,\diffsym)$ from \cref{lem:sorting}, the
diagonal elements $\transmatsym_{ii}(\shiftsym,\diffsym)$ of each column can be
used to reduce the (pseudo) degrees of the entries below them, such that all
\begin{equation}
	\rtransmatsym_{ij}(\shiftsym,\diffsym) = 
	\mtransmatsym_{ij}(\shiftsym,\diffsym) - q_{ij}(\shiftsym,\diffsym)\mtransmatsym_{jj}(\shiftsym,\diffsym)
	\label{eq:shift_elimination}
\end{equation}
for $i,j\in\{1,\dotsc,m\}, i > j$ satisfy \eqref{eq:req_degree_cond}.
Next, we discuss how to compute the required 
$q_{ij}(\shiftsym,\diffsym)$ 
to perform such a reduction and develop a modified \gls{qpld} algorithm.

\subsubsection{Quasipolynomial long division on $\realgratpolys$}%
\label{sss:tentative_recution}

Treating all entries in \eqref{eq:shift_elimination} as elements of
$\realgratpolys = \mathbb{R}(\diffsym)\left[\ggen\right]$ which allows for
rational expressions in $\diffsym$ as coefficients enables 
us to develop a tentative polynomial long division scheme in $\shiftsym$:

\begin{lem}
	\label[lem]{lem:gdiv}
	Let $x(\shiftsym,\diffsym),y(\shiftsym,\diffsym) \in \realgratpolys,\;
	y(\shiftsym,\diffsym)\neq 0$ be quasipolynomials in $\shiftsym$ with
	\emph{rational coefficients} in $\diffsym$ and 
	$\deg_\shiftsym x(\shiftsym) \ge \deg_\shiftsym y(\shiftsym)$
	as well as $\deg_{\diffsym} x(\shiftsym) \le \deg_{\diffsym} y(\shiftsym)$. 
	Then there exist $q(\shiftsym,\diffsym),r(\shiftsym,\diffsym) \in \realgratpolys$ such
	that
	\begin{equation}
		x(\shiftsym,\diffsym) = q(\shiftsym,\diffsym) y(\shiftsym,\diffsym) + r(\shiftsym,s)
		\label{eq:gdiv}
	\end{equation}
	with $\deg_\shiftsym r(\shiftsym,\diffsym) < \deg_\shiftsym y(\shiftsym,\diffsym)$,
	$\deg^+_\shiftsym r(\shiftsym,\diffsym) < \deg^+_\shiftsym y(\shiftsym,\diffsym)$
	and $\deg^-_\shiftsym r(\shiftsym,\diffsym) > \deg^-_\shiftsym y(\shiftsym,\diffsym)$.
\end{lem}
\begin{proof}
	As the following steps are valid for arbitrary coefficients, the explicit
	mention of $\diffsym$ will be omitted.
	In the general case we have 
	$\deg_\shiftsym x(\shiftsym) > \deg_\shiftsym y(\shiftsym)$
	as well as 
	$\deg^+_\shiftsym x(\shiftsym) > \deg^+_\shiftsym y(\shiftsym)$,
	$\deg^-_\shiftsym x(\shiftsym) < \deg^-_\shiftsym y(\shiftsym)$
	or both.
	Since the exponents of
	$x(\shiftsym,\diffsym) = \sum_{i=1}^u a_i(\diffsym) \shiftsym^{\alpha_i}$,
	$y(\shiftsym,\diffsym) = \sum_{i=1}^v b_i(\diffsym) \shiftsym^{\beta_i}$
	are ordered, we use
	\begin{equation}
		q(\shiftsym) =
		\begin{cases}
			\frac{a_1}{b_1}\shiftsym^{\alpha_1-\beta_1}
				, &\text{ if }
				\deg_\shiftsym^+ x(\shiftsym) > \deg^+_\shiftsym y(\shiftsym)\\
			\frac{a_u}{b_v}\shiftsym^{\alpha_u-\beta_v}
				, &\text{ if }
				\deg^-_\shiftsym x(\shiftsym) < \deg^-_\shiftsym y(\shiftsym),
		\end{cases}
		\label{eq:quotient_choice}
	\end{equation}
	where in the case that both are true, the upper branch is taken.
	Assuming the upper branch was taken, this yields 
	\begin{equation*}
	\begin{split}
		r(\shiftsym) 
		= x(\shiftsym) - q(\shiftsym) y(\shiftsym) 	
		= \sum\limits_{i=2}^{u} a_i \shiftsym^{\alpha_i}
		\!-\!\frac{a_1}{b_1}\shiftsym^{\alpha_1-\beta_1}\!
			\sum\limits_{i=2}^{v} b_i \shiftsym^{\beta_i}
	\end{split}
	\end{equation*}
	for which we find that
	\begin{equation*}
		\deg^+_\shiftsym r(\shiftsym) = \max \left\{
			\alpha_2, 
			\alpha_1 - (\beta_1 - \beta_2)
		\right\}
		<
		\deg^+_\shiftsym x(\shiftsym)
	\end{equation*}
	as well as 
	$
			\deg^-_\shiftsym r(\shiftsym) = \min \{
				\alpha_{u}, 
				\alpha_{1} - (\beta_1 - \beta_v)
		\}
	$
	in which we identify
	$\deg_\shiftsym y(\shiftsym) = \beta_1 - \beta_v$
	and use 
	$\alpha_1 = \deg_\shiftsym x(\shiftsym) + \alpha_u$
	as well as $\kappa = \deg_\shiftsym x(\shiftsym) - \deg_\shiftsym
	y(\shiftsym) \ge 0$ to obtain
	\begin{equation*}
		\begin{split}
			\deg^-_\shiftsym r(\shiftsym) = \min \{&
			\alpha_u, 
			\alpha_u + \kappa
		\}
		= \alpha_u
		= \deg^-_\shiftsym x(\shiftsym).
		\end{split}
	\end{equation*}
	This shows that the predictions in $x(\shiftsym)$ can be reduced without
	affecting the delays.
	Regarding the overall degree, we have
	$\deg_\shiftsym r(\shiftsym) = \max \{
		\alpha_2 - \alpha_u,
		\alpha_1 - \alpha_u - (\beta_1 - \beta_2)  	
	\}$
	which yields
	$\deg_\shiftsym r(\shiftsym) < \deg_\shiftsym x(\shiftsym)$.
	Hence, $x(\shiftsym)$ can be reduced in a manner such that the iteration
	\begin{equation*}
		r_{k}(\shiftsym) = r_{k-1}(\shiftsym) - q_k(\shiftsym) y(\shiftsym) 	
	\end{equation*}
	with
	$r_0(\shiftsym) = x(\shiftsym)$
	and
	$q_k(\shiftsym) =
		\frac{c_{k,1}}{b_1}\shiftsym^{\gamma_{k,1}-\beta_1}
	$
	by
	$r_k(\shiftsym) = \sum_{i=1}^{w_k} c_{k,i}\shiftsym^{\gamma_{k,i}}$
	will for some $k^+$ terminate with
	$\deg^+_\shiftsym r_{k^+}(\shiftsym) < \deg^+_\shiftsym y(\shiftsym)$
	since $x(\shiftsym)$ only has a finite number of summands.

	If required, the second branch of \eqref{eq:quotient_choice} is then taken
	and similar computations yield that restarting the iteration with 
	$r_0(\shiftsym) = r_{k^+}(\shiftsym)$
	and choosing
	$q_k(\shiftsym) =
		\frac{c_{k,1}}{b_v}\shiftsym^{\gamma_{k,1}-\beta_v}
	$
	will yield
	$\deg^-_\shiftsym r_{k^-}(\shiftsym) > \deg^-_\shiftsym y(\shiftsym)$
	after $k^-$ more steps.
	Summing up all $q_k(\shiftsym)$ from both iterations then yields the
	resulting \enquote{quotient} $q(\shiftsym)$
	with denominator $d = b_1^{k^+} b_v^{k^-}$.
\end{proof}

\subsubsection{Quasipolynomial long division on $\realgwholepolys$}%
\label{sss:actual_recution}

With \Cref{lem:gdiv}, the shift degrees have been taken care of. 
However, the coefficients of $q(\shiftsym,\diffsym)$ and
$r(\shiftsym,\diffsym)$ still are rational expressions in $\diffsym$.
Since, in the analytic setting, their denominators can be interpreted as
functions that their numerator polynomials are applied to,
the results are piecewise continuous functions which do not necessarily have
compact support.
In other words: They generally do not correspond to differential or
(distributed) delay operators with finite shift amplitude from $\entirefuncs$
and can therefore not be used to reduce the shifts as they would introduce
indefinite integrals.

Nevertheless, thanks to the Paley-Wiener theorem \cite[Thm.
7.3.1]{Hoermander1983}, if their Laplace transforms are entire functions, they
indeed are such operators.
For example, take $S_z =
\frac{1}{2}\frac{\shiftsym^z-\shiftsym^{-z}}{\diffsym} \in \realgratpolys$ from
the introduction, whose Laplace transform $\mathcal{L}\{S_z\}(s) =
\frac{1}{2s}\left(e^{sz}-e^{-sz}\right)$ appears to possess a pole at $s=0$.
However, since its numerator introduces a zero at the same location,
this pole can be lifted rendering $\mathcal{L}\{S_z\}(s)$ an entire
function and thus yielding $S_z \in \entirefuncs$ 
(which was already clear from \eqref{eq:zuazua_param_x}).
Thus, the idea is to introduce numerator zeros to $q(\shiftsym,\diffsym)$ in
form of a correction term $p(\diffsym)$ such that a modified
$q^\star(\shiftsym,\diffsym) = q(\shiftsym,\diffsym) + p(\diffsym)$ belongs to
$\realgwholepolys = \realgratpolys \cap \entirefuncs$ and is therefore
admissible.

\begin{lem}
	\label[lem]{lem:gdiv_whole}
	Let $x(\shiftsym,\diffsym),y(\shiftsym,\diffsym) \in
	\realgentirepolys,\;y(\shiftsym,\diffsym) \neq 0$ be operators from
	$\entirefuncs$, written as quasipolynomials in $\shiftsym$ with
	rational coefficients in $\diffsym$ and
	$\deg_\shiftsym x(\shiftsym) \ge \deg_\shiftsym y(\shiftsym)$, 
	as well as $\deg_{\diffsym} x(\shiftsym) \le \deg_{\diffsym} y(\shiftsym)$. 
	Then there exist 
	$q^\star(\shiftsym,\diffsym),r^\star(\shiftsym,\diffsym) \in \realgwholepolys 
	$ such that
	\begin{equation}
		x(\shiftsym,\diffsym) = q^\star(\shiftsym,\diffsym) y(\shiftsym,\diffsym) + r^\star(\shiftsym,s)
		\label{eq:gdiv_whole}
	\end{equation}
	with $\deg_\shiftsym r^\star(\shiftsym,\diffsym) = \deg_\shiftsym y(\shiftsym,\diffsym)$,
	$\deg^+_\shiftsym r^\star(\shiftsym,\diffsym) = \deg^+_\shiftsym y(\shiftsym,\diffsym)$
	and $\deg^-_\shiftsym r^\star(\shiftsym,\diffsym) = \deg^-_\shiftsym y(\shiftsym,\diffsym)$.
\end{lem}
\begin{proof}
	This proof is heavily inspired by \cite{Brethe1996,GluesingLuerssen1997}
	and only adapted to the current setup.
	First, we obtain $q(\shiftsym,\diffsym),r(\shiftsym,\diffsym) \in \realgratpolys$ by
	means of \Cref{lem:gdiv} and extract its normalised common denominator 
	$d(\diffsym) \in \realpolys$ such that 
	$q(\shiftsym,\diffsym) = \frac{\tilde q(\shiftsym,\diffsym)}{d(\diffsym)}$
	with $\tilde q(\shiftsym,\diffsym) \in \realgppolys$.
	Next, the factorisation $d(\diffsym) = \prod_{k=1}^{D} (\diffsym - s_k)$ with the roots
	$s_k$ and $D = \deg d(\diffsym)$ are used together with 
	$\hat q_0(s) = \mathcal{L}\{\tilde q(\shiftsym,\diffsym)\}$
	in the iteration 
	\begin{equation}
		\label{eq:discont_iteration}
		\hat q_{k}(s) = 
		\frac{\hat q_{k-1}(\lapsym) 
			- \hat q_{k-1}(\lapsym_k)
		}{\lapsym-\lapsym_k},
	\end{equation}
	which introduces a numerator zero for every pole from
	$q(\shiftsym,\diffsym)$ it adds.
	Thus, $q_{D}(\lapsym)$ obtained after $D$ iterations
	is an entire function on $\mathbb{C}$ and therefore
	$q^\star(\shiftsym,\diffsym) 
	= \mathcal{L}^{-1}\{q_{D}(\lapsym)\} \in \realgwholepolys$
	by \cite[Thm. 7.3.1]{Hoermander1983}.

	Now, due to the structure of \eqref{eq:discont_iteration}
	we can justify the separation
	$q^\star(\shiftsym,\diffsym) = q(\shiftsym,\diffsym) + p(\diffsym)$
	with $p(\diffsym) \in \realrationals$.
	Therefore, we obtain the degree conditions from \Cref{lem:gdiv}
	with equal signs for the new remainder $r^\star(\shiftsym,\diffsym)$.
\end{proof}

\subsubsection{Algorithm}%
\label{sss:reduction_algorithm}

Equipped with a suitable way to reduce the degrees of its elements by 
means of the \gls{qpld} from \Cref{lem:gdiv_whole},
we now reduce the degrees in the entire matrix $\transmat(\shiftsym,\diffsym)$
by applying a series of row transformations.

\begin{algorithm}
	\caption{Shift reduction in $\transmat(\shiftsym,\diffsym)$}
	\label[alg]{alg:shift_red}
	\begin{algorithmic}[1]
		\Function{build\_transform}{$\mat X,j$}
			\State $\mat L_c \gets \mat I_{\dimfwd}$
			\For{$i \in \{j+1,\dotsc,\dimfwd\}$}
			\State{// From \Cref{lem:gdiv_whole}:}
			\State $q^\star_{i,j}
			\gets$ \Call{qpld}{$x=\mat X[i,j],y=\mat X[j,j]$}
			\State $\mat L_c[i,j] \gets -q^\star_{i,j}$
			\EndFor
			\State\Return $\mat L_c$
		\EndFunction
		\State
		\Procedure{reduce\_shifts}{\transmat}
			\State $\bar\transmat \gets \transmat,\;
				\mat L \gets \mat I_{\dimfwd}$
				\While{$\sum_{i=0}^{\dimfwd}\cdeg_i \rtransmat > \totalshift$}
				\If{$\deg_\shiftsym \rtransmatsym_{ii} 
					\neq\deg_\shiftsym \rtransmatsym_{jj}
					 \;\forall i,j \in \{1,\dotsc,\dimfwd\}
					$}
					\For{$j \in \{1,\dotsc,\dimfwd-1\}$}
						\State $\mat L_c \gets$ 
							\Call{build\_transform}{$\bar\transmat, j$}
						\State $\bar\transmat \gets \mat L_c \transmat$,
							$\mat L \gets \mat L_c \mat L$
					\EndFor
				\Else
					\State $\mat L_c \gets$ 
						\Call{mqpld}{$\rtransmat$}
					\State $\bar\transmat \gets \mat L_c \transmat$,
						$\mat L \gets \mat L_c \mat L$
				\EndIf
			\EndWhile
			\State\Return $\bar\transmat$,$\mat L$
		\EndProcedure
	\end{algorithmic}
\end{algorithm}
\begin{lem}
	The matrix $\bar\transmat(\shiftsym,\diffsym) 
	= \mat L(\shiftsym,\diffsym) \transmat(\shiftsym,\diffsym)$
	obtained by \Cref{alg:shift_red} satisfies the requirements
	\eqref{eq:req_alg_conds}.
\end{lem}
\begin{proof}
	In the while loop, for the case of unique pivot degrees, the degrees in
	each column are corrected separately from left to right.
	Therein, each correction step can introduce new shifts to the rows below its
	pivot element.
	In the case of coinciding degrees of pivot elements, a \gls{qpldm} is
	performed which reduces the degrees in several columns at once but is not
	discussed here due to space limitations.
	Then the loop repeats, thus, effectively reducing the delays in
	each column by the (nonzero) degree differences of the pivot elements until the
	desired state dimension is reached.
\end{proof}

\subsection{Deriving a hyperbolic MIMO controller form}%

Applying \Cref{alg:shift_red} to $\transmat(\shiftsym,\diffsym)$ from \eqref{eq:io_form} we obtain
\begin{equation}
	\mat L^{-1}(\shiftsym,\diffsym) \oinput(t)
	= \bar\transmat(\shiftsym,\diffsym) \flatout(t)
	\label{eq:io_form_transformed}
\end{equation}
with $\hat{\transmat} = \lccm_{\diffsym,\shiftsym} \bar\transmat(\shiftsym,\diffsym)$
having full rank.
Thus, by means of the separation 
$\rtransmat(\shiftsym,\diffsym) 
= \ltransmat
	\mat K(\shiftsym,\diffsym)
	+ \ttransmat(\shiftsym,\diffsym)
$
where
$\mat K(\shiftsym,\diffsym) = \diag(
\diffsym^{\kron_1}\shiftsym^{\predsym_{1}},
\dotsc,
\diffsym^{\kron_m}\shiftsym^{\predsym_{m}}
)$
with the highest derivative and prediction orders 
$\kron_i = \deg_{\diffsym} \bar \transmatsym_{ii}(\shiftsym,\diffsym)$
and 
$\predsym_{i} = \deg^+_\shiftsym \bar \transmatsym_{ii}(\shiftsym,\diffsym)$
we obtain the explicit \gls{fde} system
\begin{equation}
	\mat K(\shiftsym,\diffsym) \flatout(t) 
	\!=\!
	\hat{\transmat}^{-1}\!\!
	\left(
		\mat L^{-1}(\shiftsym,\diffsym) \oinput(t)\!
		- \tilde{\transmat}(\shiftsym,\diffsym) \flatout(t)
		\!
	\right)
	\label{eq:dde_system}
\end{equation}
which reads
\begin{equation}
	\flatsym_i^{(\kron_i)}(t+\predsym_i) = 
	- \sum_{j=1}^m \left(
		 (\tilde a^\star_{ij} \flatsym_j)(t)
		 - (\tilde b^\star_{ij} \oinputsym_j)(t)
		 \right)
	\label{eq:dde_system_analytic}
\end{equation}
for $i=\{1,\dotsc,m\}$ in the analytic setting.
Herein, the operators 
$\tilde a^\star_{ij}$ and $\tilde b^\star_{ij}$
by construction have compact support,
which is why their application yields
definite integrals whose limits depend on the degrees of
$\tilde{\transmat}(\shiftsym,\diffsym)$ and
$\mat L^{-1}(\shiftsym,\diffsym)$ w.r.t.~$\shiftsym$, respectively.
Thus, \eqref{eq:dde_system_analytic} states a well-posed initial value problem
if the components $\flatsym_i(t+\tau)$ of the flat output are given for $\tau
\in [\delaysym_i,\predsym_i]$ with $\delaysym_i = \deg^-_\shiftsym
\rtransmatsym_{ii}(\shiftsym,\diffsym)$ and the input $\oinput(t)$ is known for
$t \in \reals$.

Hence, we introduce the \gls{hcf} state
\begin{subequations}
\begin{align}
	\nfodesym_{i,j}(t) &\coloneq \flatsym_i^{(j-1)}(t),
	\; j \in \{1,\dotsc,\kron_i\}\\
	\nfpdesym_{i}(\nfspatsym, t) &\coloneq 
		\flatsym_i^{(\kron_i)}(t+\delaysym_i+\nfspatsym)
	,\;\nfspatsym\in[\nfleftbound{i},\nfrightbound{i}]
\end{align}%
\label{eq:hcf_state_def}%
\end{subequations}
with $\nfrightbound{i}=\predsym_i - \delaysym_i$
for $i\in\{1,\dotsc,m\}$ which is equivalent to the flat output on these intervals.
Substituting\footnote{%
	Some expressions in $\flatout(t)$ may remain as e.g.~the values
	$\flatsym_i^{(d)}(t+\nfspatsym) $ are not part of the \gls{hcf} state
	$\nfstate(\cdot,t)$ for $d < \kron_i$, cf.~\eqref{eq:mimo_hcf}.
	However, these expressions can be easily eliminated by means 
	shown in \cite[Lem. 5]{Gehring2023}.
}
the definitions \eqref{eq:hcf_state_def} in \eqref{eq:dde_system_analytic}
then allows to identify the operators 
$a^\star_{ij}, b^\star_{ij}$ as well as weights $\bm a_{ij}$ 
in \eqref{eq:mimo_hcf}, yielding the \gls{hcf} of system \eqref{eq:orig_sys}.
\begin{lem}
	The type of \gls{hcf} obtained %
	is a quasi controller form, (cf.~\Cref{sss:gform_pred_and_delay}), 
	with only input predictions and delays but no input derivatives.
\end{lem}
\begin{proof}
	The proof of \Cref{lem:gdiv} yields $\deg_{\diffsym} q(\shiftsym, \diffsym)
	< 0$ due to the structure of $d(\diffsym)$. 
	Then by \Cref{lem:gdiv_whole}, $q^\star(\shiftsym, \diffsym)$ has
	non-positive degree w.r.t.~$\diffsym$ as do all entries in $\mat L_i(\shiftsym, \diffsym)$
	from \Cref{alg:shift_red}.
	Now, as the inverse of the triangular $\mat L(\shiftsym, \diffsym)$
	is computed by
	forward substitution, all $\bar q^\star_{ij}(\shiftsym, \diffsym)$ 
	therein are linear combinations of the
	$q_{ij}^\star(\shiftsym, \diffsym)$, hence,
	$\deg_{\diffsym} \bar q^\star_{ij}(\shiftsym, \diffsym) \le 0$.
	Therefore, the application of the $\tilde b^\star_{ij}$ yields 
	(distributed) shifts but no derivatives.
\end{proof}

\section{Reduction of the benchmark system}%
\label{sec:example}

Consider again the system from \Cref{ssc:benchmark}
with $\leftbound = \pi$, $\rightbound = 10$ and $\mass = 1$
substituted into \eqref{eq:zuazua_param_u} yielding
\begin{equation*}
	\transmat(\shiftsym,\diffsym) 
	=
	\frac{1}{4} 
	\begin{pmatrix}
		a_1(\diffsym)\shiftsym^{\pi} + a_2(\diffsym)\shiftsym^{-\pi}
			& -\shiftsym^{\pi} -\shiftsym^{-\pi} 
			\\
			b_1(\diffsym)\shiftsym^{10} + b_2(\diffsym)\shiftsym^{-10}
			& \shiftsym^{10} +\shiftsym^{-10} 
			\\
	\end{pmatrix}
\end{equation*}
with $a_1(\diffsym) = b_1(\diffsym) = \diffsym^2 + 2\diffsym$ and 
$a_2(\diffsym) = b_2(\diffsym) = \diffsym^2 - 2\diffsym$
for which
$
\hat\transmat = \lccm_{\shiftsym,\diffsym} \transmat(\shiftsym,\diffsym)
=
\lccm_{\diffsym,\shiftsym} \transmat(\shiftsym,\diffsym) 
$ 
is singular as shown in \eqref{eq:zuazua_lccm} which is why the degree
conditions \eqref{eq:req_degree_cond} from \Cref{sss:reduction_requirements} fail;
e.g.
$\deg_\shiftsym \transmatsym_{21}(\shiftsym,\diffsym) 
= 20 > \deg_\shiftsym \transmatsym_{11}(\shiftsym,\diffsym) = 2\pi$.

Continuing with 
\Cref{sss:preliminary_transforms}, 
no superfluous delays are present,
rendering the delay triangularisation and homogenisation obsolete.
Furthermore, the row degrees w.r.t.~$\shiftsym$ are already ordered
and the columns do not need to be reordered as
their entries satisfy the degree condition \eqref{eq:req_derivative_ordering}.

Therefore, by means of \Cref{lem:gdiv} %
the quotient
\begin{equation*}
\scale[0.9]{
	\begin{aligned}
		q_{21}&(\shiftsym,\diffsym) \!=\! 
		\sigma^{10 - \pi}
		\!+\! \frac{2 - \diffsym}{\diffsym + 2}\sigma^{10 - 3 \pi} 
		\!-\! \frac{\diffsym + 2}{\diffsym - 2}\sigma^{-10 + 3 \pi}
		\!+\! \sigma^{-10 + \pi} 
	\end{aligned}
}
\end{equation*}
is obtained, where one iteration was needed for each case
in \eqref{eq:quotient_choice}.
Next, \Cref{lem:gdiv_whole} yields the correction term
$p(\diffsym) = \frac{16 e^{-20 + 6\pi}}{\left(\diffsym^{2} - 4\right)}$
such that $q_{21}^\star(\shiftsym,\diffsym) \in \realgwholepolys$.
Using the latter in \eqref{eq:shift_elimination} then yields
the reduced element
\begin{equation*}
\scale[0.9]{
	\begin{aligned}
	\bar \transmatsym_{21}(\shiftsym,\diffsym) =
	\!- \frac{16 \diffsym}{\left(\diffsym - 2\right) e^{20 - 6 \pi}}\sigma^{\pi}
	\!+ \frac{\diffsym\left(\diffsym^{2} + 4 \diffsym + 4\right)}{\diffsym - 2}\sigma^{-10 + 4 \pi}\\
	\!+ \frac{\diffsym \left(\diffsym^{2} - 4 \diffsym + 4\right)}{\diffsym + 2} \sigma^{10 - 4 \pi}
	\!- \frac{16 \diffsym}{\left(\diffsym + 2\right) e^{20 - 6 \pi}}\sigma^{-\pi}.
	\end{aligned}
}
\end{equation*}
with 
$\deg_\shiftsym \bar\transmatsym_{21}(\shiftsym,\diffsym) 
= \deg_\shiftsym \transmatsym_{11}(\shiftsym,\diffsym)$.
As the requirements for the second column are already satisfied, this concludes
the reduction and we obtain
$\ltransmat = \lccm_{\diffsym,\shiftsym}\rtransmat(\shiftsym,\diffsym) = \diag(1, -2)$.
\footnote{%
The complete computation of the \gls{hcf} for the example with annotated steps can be found at
\url{https://github.com/umit-iace/iace-pub-2026-ecc}.
}

\section{Concluding remarks}%
\label{sec:outlook}

This contribution establishes a \gls{hcf} for \acs{mimo} systems with
dynamic boundaries and proposes a flatness based algorithm to compute
said form for a specific class of systems.
While the algebraic approach taken to obtain this result is limited
to linear systems, it enables a systematic procedure which
lends well to an implementation using computer algebra.
However, the insight gained from an analytic interpretation of the
algebraic procedure allows for the technique to be extended to nonlinear boundary
systems, which will be discussed in a forthcoming publication.
Furthermore, the obtained \gls{hcf} contains input delays, which will result in
a dynamic extension, thus further analysis is required to obtain a form without
delays as this would enable the application of quasi-static state feedback.
Finally, the system class \eqref{eq:orig_sys} is lacking in-domain couplings.
As this class was recently addressed in \cite{Schmidt2025} 
without boundary dynamics, the next logical step is to combine the two approaches.

\bibliographystyle{IEEEtran}
\bibliography{IEEEabrv,./common/bib/abbreviations.bib,./common/bib/sources.bib}

\end{document}